\documentclass[12pt]{article}

\usepackage{amsmath,amsthm,amssymb}
\usepackage{setspace}
\usepackage{mathptmx}
\usepackage{subcaption}
\usepackage{bm}
\usepackage{tikz}
\usepackage{algorithm2e}
\pgfdeclarelayer{background}
\pgfdeclarelayer{foreground}
\pgfsetlayers{background,main,foreground}
\usepackage{url}

\theoremstyle{plain}
\newtheorem{theorem}{Theorem}

\newcommand{\lmaj}{\ensuremath{\succcurlyeq}}
\newcommand{\rmaj}{\ensuremath{\preccurlyeq}}
\newcommand{\lsmaj}{\ensuremath{\succ}}

\newcommand{\seq}[1]{\ensuremath{\left( #1 \right)}}

\title{
One-Pass Graphic Approximation of Integer Sequences 
\footnote{Official contribution of the National Institute of Standards and Technology; not subject to copyright in the United States.}
}

\author{Brian~Cloteaux \\
\small National Institute of Standards and Technology,\\
\small Applied and Computational Mathematics Division,\\
\small Gaithersburg, MD \\
\texttt{brian.cloteaux@nist.gov}}

\date {}

\begin{document}

\maketitle

\begin{abstract}
A variety of network modeling problems begin by generating a degree
sequence drawn from a given probability distribution. 
If the randomly generated sequence is not graphic, we give a new
approach for generating a graphic approximation of the sequence.  This
approximation scheme is fast, requiring only one pass through the
sequence, and produces small probability distribution distances for
large sequences.
\end{abstract}

\section{Introduction}
Model creation of real-world complex networks, such as social and
biological interactions, remains an important area of research.
The modeling and simulation of these systems often
requires developing realistic graph models for the underlying networks
embedded in these systems.
A common first step in the creation of these graph models is to
construct a graph with a given degree sequence. If we can construct a
graph from an integer sequence, we say that the sequence is graphic. 

One problem from the random modeling of graphs is selecting a
graphic degree sequence from some given probability distribution.
More specifically, how do we deal with a randomly drawn integer sequence that
is not graphic?
There have been two approaches to this problem. The first is to simply discard
the sequence and repeatedly select a new sequence until a graphic
sequence is found.  This is the approach used
by the NetworkX~\cite{Hagberg:2008} graph library.  A disadvantage to this
approach is that for some probability distributions, the chance of
selecting a graphic sequence is small.  Thus the NetworkX library
places an upper limit on the number of attempts to generate a graphic
sequence before abandoning the task.  This is  potentially a 
computationally expensive task for very large sequences selected from a
sequence with a small probability of being graphic.

In response to this difficulty, a second approach,
suggested by Mihail and Vishoi~\cite{Mihail:2002},
is to find the closest graphic sequence to the original non-graphic
degree sequence using a given distance measure.  Using a distance measure
that they called the discrepancy, Mihail and Vishoi
provided a polynomial-time algorithm for finding these graphic approximations by
reducing this problem to maximum cardinality matching;
unfortunately, their algorithm is not practical for many cases.
They posed the question whether a practical algorithm for
determining nearest graphical sequence under the discrepancy measure could be
found.  This question was answered by Hell and
Kirkpatrick~\cite{Hell:2009} who
introduced two algorithms for computing a minimal discrepancy graphic
sequence in lower polynomial time.

When Mihail and Vishoi considered this problem,
they posed a couple of natural questions.  The first is "What is a natural
notion of distance between two degree sequences?", and the second is "Are
there efficient algorithms to approximate a non-graphic sequence with a
graphic sequence whose distance is small?"  We examine both of these
questions and make a new contribution to the problem by introducing a
new approach that requires only one pass through the sequence in order
to give a graphic approximation. Further, we suggest that 
that minimizing distances between probability distributions,
specifically the total
variation distance, is a more appropriate approach to this problem.
Finally, we show that under the total variation distance,
the quality of the result from our
approximation algorithm for many sequences improves as the size of the
sequences increases.

\section{Basic  Results}
Before introducing the algorithm, we need some preliminary
definitions and results.
A {\it degree sequence} $\alpha = \seq{\alpha_1, \alpha_2, ..., \alpha_{n}} $
is a set of non-negative integers such that $\alpha_1 \geq \alpha_2 \geq
... \geq \alpha_n \geq 0$.  As a notational convenience, if a value $a$
is represented $r$ times in a sequence we represent that subsequence as
$a^r$, i.e.,
\begin{equation}
\seq{\underbrace{a,...,a}_{r}} = \seq{a^r}.
\end{equation}
If there exists a graph $G$ whose node degrees match the sequence $\alpha$, then
$\alpha$ is said to be {\it graphic}; else, the sequence is called {\it
non-graphic}.

If a sequence $\alpha$ has an even sum $s$ where $s \leq
|\alpha|(|\alpha| -1)$, then we say that the sequence is {\it potentially
graphic}.  It is straightforward to see that any sequence that violates
this condition is not graphic. Conversely, if this condition is met, then
there exists some sequence having length $|\alpha|$ and sum $s$ that is
graphic (this follows from Chen~\cite{Chen:1988}, Lemma 1).
The minimal requirement for graphic approximation of a  
non-graphic sequence is that original sequence must be potentially graphic.

A tool we use to work with degree sequences is the
partial ordering called {\it majorization}.
A degree sequence $\alpha$ majorizes (or dominates) the
integer sequence $\beta$,
denoted by $\alpha \succcurlyeq \beta$, if for all $k$ from $1$ to $n$
\begin{equation} \label{eqn:major1}
\sum_{i=1}^{k} \alpha_i \geq \sum_{i=1}^{k} \beta_i,
\end{equation}
and if the sums of the two sequences are equal. If $\alpha \lmaj \beta$
and there exists an index $k$ where
$\sum_{i=1}^{k} \alpha_i > \sum_{i=1}^{k} \beta_i$, then we say that
$\alpha$  {\it strictly} majorizes $\beta$, denoted as $\alpha \lsmaj
\beta$.

Over the set of integer partitions $\mathcal{L}_p$
for some positive integer $p$,
majorization forms the lattice $(\mathcal{L}_p, \lmaj)$
\cite{Brylawski:1973} where the meet operator is defined as
follows: 
\begin{equation}
\begin{split}
\label{eqn:meet-operation}
( \alpha \wedge \beta )_k & = \min\lbrace \sum_{i=1}^{k} \alpha_i,
\sum_{i=1}^{k} \beta_i \rbrace - \sum_{i=1}^{k-1} (\alpha \wedge
\beta)_i \\
        & =
\min\lbrace \sum_{i=1}^{k} \alpha_i, \sum_{i=1}^{k} \beta_i \rbrace
- \min\lbrace \sum_{i=1}^{k-1} \alpha_i, \sum_{i=1}^{k-1} \beta_i
\rbrace, 
\end{split}
\end{equation}

For our purposes, the usefulness of comparing degree sequences using
majorization stems from the following result which shows that
every sequence that is majorized by a
graphical sequence must also be graphical.

\begin{theorem}[Ruch and Gutman~\cite{Ruch:1979}, Theorem 1]
\label{thm:graphicality}
If the degree sequence $\alpha$ is graphic and $\alpha \lmaj \beta$,
then $\beta$ is graphic.
\end{theorem}

This result  establishes the location of the graphic sequences in the
lattice $(\mathcal{L}_p,\lmaj)$.
If the integer $p$ is even, then there will be a small number of elements
at the bottom of the partition lattice $\mathcal{L}_p$ that are graphic,
while the remaining majority of the partitions will be non-graphic
\cite{Pittel:1999}.  It also shows that the graphic sequences form a
semi-lattice in $\mathcal{L}_p$.

A degree sequence which has precisely one labeled realization is
called a {\it threshold sequence} and the resulting
realization is called a {\it threshold graph} \cite{Mahadev:1995}.
The following theorem shows that the threshold sequences sit at the top of the
semi-lattice of graphic sequences.

\begin{theorem}[\cite{Mahadev:1995}, Theorem 3.2.2] \label{thm:top_graphic}
A graphic degree sequence $\alpha$ is threshold if and only if there does
not exist a graphic sequence $\beta$ such that $\beta \lsmaj \alpha$.
\end{theorem}

This next theorem provides a decomposition result that we can use to recognize
threshold graphs (and, by extension, threshold sequences).

\begin{theorem}[\cite{Mahadev:1995}, Theorem 1.2.4] \label{thm:dom_decomp}
A graph $G$ is threshold if and only if it can be constructed from a
one-vertex by repeatedly adding either an isolated or dominating vertex.
\end{theorem}

\section{A One-Pass Approximation Scheme}

Our approximation scheme for the sequence $\alpha$ first involves creating a
graphic sequence $T$ with the same length and sum as $\alpha$ and that has 
several desirable features.  This sequence $T$ is defined as follows:

For the value $m = \min \lbrace k | s \leq k \cdot (k-1) \rbrace$,
\begin{equation*} T(n,s) =
\begin{cases}
\seq{0^{n}} & \text{if $s=0$,}\\
\seq{\gamma_1, ..., \gamma_m, 0^{n-m}} \text{ where $\gamma = T(m,s)$}
& \text{if $n>m$,}\\ 
\seq{n-1,\gamma_1+1,...,\gamma_{n-1}+1} \text{ where $\gamma =
T(n-1,s-2\cdot(n-1))$} & \text{if $n \leq m$.} \\
\end{cases}
\end{equation*}

By comparing the definition of $T$ with Theorem \ref{thm:dom_decomp},
we see that $T(n,s)$ defines graphic threshold sequence.
Thus, $T(n,s)$ is a maximal graphic sequence in the majorization lattice
(Theorem \ref{thm:top_graphic}). At the same time,  from the definition, this
sequence contains the minimum possible number of non-zero elements for a
graphic sequence.

Additionally, an advantage in using the sequence $T$ is that
we do not have to store or pre-compute this sequence.
We can quickly compute any index of $T(n,s)$ in constant time as shown
by the following result.

\begin{theorem}
\label{thm:equiv_def}
The sequence $T(n,s)=\seq{(p+1)^{q},p^{p+1-q},q,0^{n-p+2}}$ where
\begin{equation}
p = \left\lfloor \frac{\sqrt{4s+1}-1}{2} \right\rfloor,
\end{equation}
and
\begin{equation}
	q = \frac{s - p(p+1)}{2}.
\end{equation}
\end{theorem}

\begin{proof}
We show this result by induction. For the base step, if $s=2$ then
we have the sequence $\seq{1^2,0^{n-2}}$ which matches the above
definition.

For the inductive step, we assume that for all sums $<s$,
the above result holds.  Take the sequence $\gamma = T(n,s)$ where
without a loss of generality we assume that $\gamma_n > 0$.
By reducing the sequence by removing the value $\gamma_1$ and
subtracting one from the remaining values,
we create the  new sequence $\gamma' = \seq{\gamma_2-1,...,\gamma_n-1}$
whose  sum is $s' = s - 2n$.
From the inductive hypothesis,  this new sequence can be written as
%If $|\gamma'| = n-1$ s then
$\gamma'=\seq{(p'+1)^{q'},p'^{p'+1-q},q'}$ where $p'= n-2$. 
By adding back the value $\alpha_1$ along with adding one to the
values in $\gamma'$, we can now write $\gamma$ as
$\gamma=\seq{(p+1)^{q},p^{p+1-q},q}$ where
$p = p'+1$, $q = q'+1$, and $s' = s-2(p+1)$.

We now consider the values for  $p$ and $q$ in $\gamma$.
The value of the parameter $q$ can be computed as follows:
\begin{equation}
\label{eqn:newq}
\begin{split}
q &= q'+1 = \frac{s' - p'(p'+1)}{2} + 1  \\
  &= \frac{(s-2(p+1)) - (p-1)p}{2} + 1 = \frac{s - p(p+1)}{2}  \\
\end{split}
\end{equation}
The parameter $p = \left\lfloor \frac{\sqrt{4s+1}-1}{2} \right\rfloor$
is equivalent to computing the integer $p$ where
$p(p+1) \leq s < (p+1)(p+2)$.  
Equation \ref{eqn:newq} establishes that $p(p+1) \leq s$.  
The second part of the inequality holds as follows:
\begin{equation}
\begin{split}
s &= s'+2(p+1) = s' + 2(p'+2) = s'+2p'+4 \\
  &< s'+4p'+6 = (p'+1)(p'+2)+4p'+6 \\
  &= (p'+2)(p'+3) = (p+1)(p+2) \\
\end{split}
\end{equation}
Thus the premise is established.

\end{proof}

We apply this result to create the approximation algorithm.  The
idea is to take the meet operation between
the original non-graphic sequence $\alpha$ and the sequence $T$, where
the length and sum of $\alpha$ and $T$ are equal.
Since $\alpha \wedge T \rmaj T$, then by Theorem~\ref{thm:graphicality}
the resulting
sequence will be graphic. Also, since we are able to compute each value of
$T$ in constant time, we only need one-pass through the sequence
$\alpha$ to create the graphic approximation for a run-time of $O(n)$.
The full algorithm is shown in Algorithm~\ref{alg:graphic_approx}.

%\begin{figure}
%\centering
\begin{algorithm}[H]
\singlespacing
\SetAlgoLined
\DontPrintSemicolon
\KwIn{$\alpha$ - a potentially graphic sequence where $\alpha_1 \geq
\alpha_2 \geq ... \geq \alpha_n \geq 0$,\\
s - sum of sequence $\alpha$}
\KwOut{$\beta$ - a graphic sequence}
$p \gets \left\lfloor \frac{\sqrt{4s+1}-1}{2} \right\rfloor$ \;
$q \gets \frac{s - p(p+1)}{2}$ \;
$\beta \gets \seq{0^{|\alpha|}}$, $as \gets 0$, $bs \gets 0$, $ts \gets 0$ \;
\For {$i \gets 1$ \KwTo $|\alpha|$}{
	$as \gets as + \alpha[i]$ \;
	$ti \gets 0$ \;
	\uIf {$i \leq q$}{
		$ti \gets p+1$\;
	}\uElseIf {$i > q \And i \leq p+1$}{
		$ti \gets p$ \;
	}\uElseIf {$i = p+2$}{
		$ti \gets q$ \;
	}
	$ts \gets ts + ti$ \;
	$\beta[i] \gets \min\{as,ts\} - bs$ \;
	$bs \gets bs + \beta[i]$ \;
}
\Return $\beta$ \; 
\caption{Graphic Sequence Approximation Algorithm}
\label{alg:graphic_approx}
\end{algorithm}
%\end{figure}

We point out a couple of practical points about this algorithm.
While a degree sequence is being generated, it can be sorted simultaneously
by using a binned sort.
The sorted sequence can be tested whether it is
graphic or not in an additional $O(\sqrt{s})$ steps \cite{Cloteaux:2015}. 
If the sequence is not graphic, this approximation scheme only requires
one more pass through the sequence to produce a graphic sequence.
In addition, if the sequence sum $s$ is not even, then we can use this
approximation scheme for the sum $s-1$ to produce a graphic sequence. 

\section{Comparing the Resulting Sequences}

Mihail and Vishoi originally suggested the discrepancy metric
for measuring the distance between the original sequence and its
approximation.
The discrepancy, $\Delta$ measures differences between individual values
in the two sequences, i.e., for the sequences $\alpha$ and $\beta$,
\begin{equation}
\Delta(\alpha,\beta) = \sum_{i} |\alpha_i - \beta_i|.
\end{equation}
As mentioned before, there are several polynomial-time algorithms 
that can find minimal discrepancy graphic approximations
\cite{Mihail:2002,Hell:2009}.

For the problem of drawing a graphic sequence from a probability distribution,
we suggest that a better measure to minimize between two sequences
is the distance between the probability distributions of the two
sequences.  We write the sequence $\alpha = 
\seq{1^{r_1}, 2^{r_2}, ..., (n-1)^{r_{n-1}}}$ where all $i$ where $0
\leq r_i \leq n-1$, and so we are interested in a minimal
matching to the vector $\seq{r_1, r_2, ..., r_{n-1}}$ by the
probability distribution of the approximation sequence.
The probability of an individual value $i$ being chosen is $P(\alpha)_i =
r_i/n$, and we denote the probability distribution of a sequence
$\alpha$ as $P(\alpha) = \seq{r_1/n, r_2/n, ..., r_{n-1}/n}$.
For measuring probability distances for creating approximations, we use the
{\it total variation distance}.
This distance measures the
largest possible probability difference for an event between two
distributions $P$ and $Q$ and is defined as
\begin{equation}
	TV(P,Q) = \sup_{\omega \in \Omega} |P_\omega -Q_\omega|.
\end{equation}

For the general case of this problem, the time needed to compute a minimum
total variation graphic sequence for a given non-graphic sequence is an
open problem.
We show that for many useful distributions, the distribution arising
from the approximation sequence created by Algorithm
\ref{alg:graphic_approx} will asymptotically approach a minimal total
distance. In particular, this applies to sequences where the sum $s$ of the
sequence is $\sqrt{s} = o(n)$.
These sequences include many commonly used distributions, such as power-law
distributions.  

\begin{figure}
\centering
\includegraphics[width=.9\textwidth]{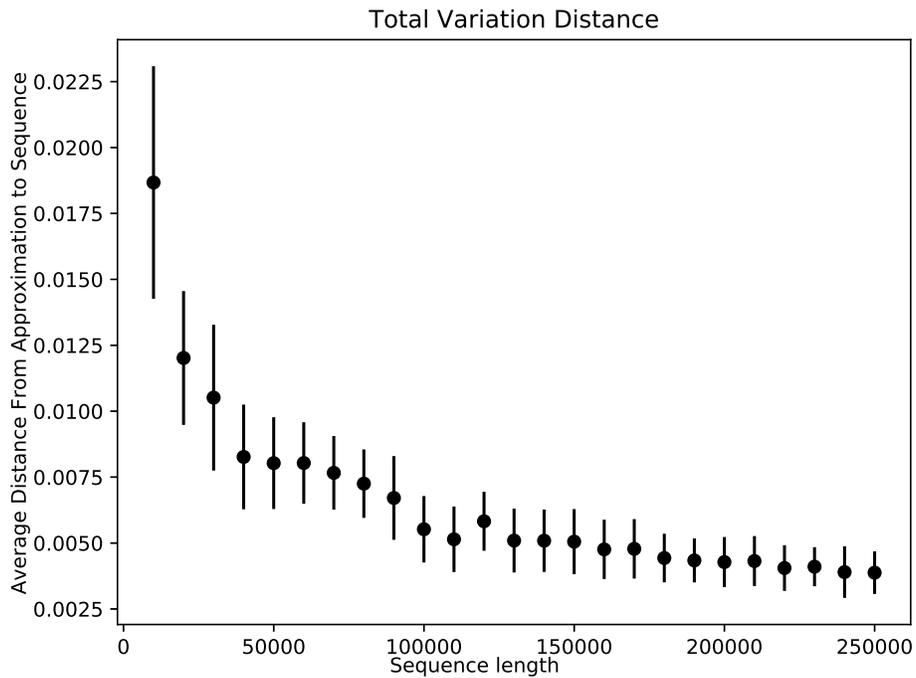}
\caption{This figure represents the average difference between a
non-graphic sequence selected from a power-law distribution with
exponent of 2 and its approximation created from Algorithm
\ref{alg:graphic_approx}. Each point represents the average distance for 30
sequences created at a given length. 
The error bars represent one standard deviation.
}
\label{fig:approx-power} 
\end{figure}

\begin{theorem}
Let $\mathcal{F}$ be a family of integer sequences where for any $\alpha
\in \mathcal{F}$ whose sum is $s$ and length is $n$, then $s = o(n^2)$.
For $\alpha \in \mathcal{F}$,
\begin{equation}
	\lim_{n\rightarrow\infty} TV(P(\alpha),P(\alpha \wedge T(n,s))) =0.
\end{equation}
\end{theorem}

\begin{proof}
From the  definition of the sequence $T$, the number of non-zero entries
it contains is $m$.  From the
definition of the meet operation, the values for the sequence $T\wedge
\alpha$ will be equal to $\alpha_i$ for all indices $i > m$.
It follows from Theorem \ref{thm:equiv_def}, that 
a simple upper bound on $m$ is $\sqrt{s}$. Thus an upper
bound on the possible change in the probability of any one value in the
distribution is
\begin{equation}
	\sup_{\omega} |P(\alpha)_\omega -P(\alpha \wedge T(n,s))_\omega| \leq \frac{m}{n} <
\frac{\sqrt{s}}{n}.
\end{equation}
Since $s = o(n^2)$, then it follows
\begin{equation}
	\lim_{n\rightarrow\infty} TV(P(\alpha),P(\alpha \wedge T(n,s))) \leq
	\lim_{n\rightarrow\infty} \frac{\sqrt{s}}{n} = 0,
\end{equation}
completing the argument.
\end{proof}

This convergent behavior
is shown in Figure \ref{fig:approx-power}.
In this example, for each point, 30 non-graphic sequences with even sum were
drawn from a power-law distribution with exponent 2 and the average
distance to their approximation was computed.  
We see that the average total variation distance between the original
sequences and their approximations approach zero as the sequence length grows.

\section{Conclusions}

We give a simple and extremely fast method for approximating a
non-graphic sequence with a graphic one.  This approach is particularly 
well-suited for large sequences, since only one-pass through the
sequence is needed, and the probability distances are reduced as the
sequences become larger.  This presents a practical algorithm for
large network modeling applications.

{ \singlespacing
\bibliographystyle{nist}
\bibliography{fast_graphic}

\begin{thebibliography}{1}
\newcommand{\enquote}[1]{``#1''}
\providecommand{\url}[1]{\texttt{#1}}
\providecommand{\urlprefix}{ }
\providecommand{\eprint}[2][]{\url{#2}}

\bibitem{Hagberg:2008}
A.~A. Hagberg, D.~A. Schult, and P.~J. Swart (2008).
\newblock \enquote{Exploring network structure, dynamics, and function using
  {NetworkX}}.
\newblock In Proceedings of the 7th Python in Science Conference (SciPy2008),
  pp. 11--15. Pasadena, CA USA.

\bibitem{Mihail:2002}
M.~Mihail and N.~Vishnoi (2002).
\newblock \enquote{On generating graphs with prescribed degree sequences for
  complex network modeling applications}.
\newblock In Proceedings of the 3rd Workshop on Approximation and Randomization
  Algorithms in Communication Networks (ARACNE), pp. 1--11.

\bibitem{Hell:2009}
P.~Hell and D.~Kirkpatrick (2009).
\newblock \enquote{Linear-time certifying algorithms for near-graphical
  sequences}.
\newblock \emph{Discrete Math.} \textbf{309}~(18), 5703--5713.
\newblock \urlprefix\url{http://dx.doi.org/10.1016/j.disc.2008.05.005}.

\bibitem{Chen:1988}
Y.~C. Chen (1988).
\newblock \enquote{A short proof of {K}undu's {$k$}-factor theorem}.
\newblock \emph{Discrete Math.} \textbf{71}~(2), 177--179.
\newblock \urlprefix\url{http://dx.doi.org/10.1016/0012-365X(88)90070-2}.

\bibitem{Brylawski:1973}
T.~Brylawski (1973).
\newblock \enquote{The lattice of integer partitions}.
\newblock \emph{Discrete Math.} \textbf{6}, 201--219.
\newblock \urlprefix\url{http://dx.doi.org/10.1016/0012-365X(73)90094-0}.

\bibitem{Ruch:1979}
E.~Ruch and I.~Gutman (1979).
\newblock \enquote{The branching extent of graphs}.
\newblock \emph{Journal of Combinatorics, Information, \& System Sciences}
  \textbf{4}~(4), 285--295.

\bibitem{Pittel:1999}
B.~Pittel (1999).
\newblock \enquote{Confirming two conjectures about the integer partitions}.
\newblock \emph{Journal of Combinatorial Theory, Series A} \textbf{88}~(1),
  123--135.
\newblock \urlprefix\url{http://dx.doi.org/10.1006/jcta.1999.2986}.

\bibitem{Mahadev:1995}
N.~Mahadev and U.~Peled (1995).
\newblock Threshold Graphs and Related Topics.
\newblock North Holland.
\newblock {A}nnals of Discrete Mathematics 56.

\bibitem{Cloteaux:2015}
B.~Cloteaux (2015).
\newblock \enquote{Is this for real? {F}ast graphicality testing}.
\newblock \emph{Computing in Science \& Engineering} \textbf{17}~(6), 91--95.
\newblock \urlprefix\url{http://dx.doi.org/10.1109/MCSE.2015.125}.

\end{thebibliography}
}

\end{document}